\documentclass[12pt, 
]{article}
\usepackage{booktabs}

\usepackage[utf8]{inputenc}
\usepackage{amssymb}
\usepackage{amsmath}

\usepackage{tikz}
\usepackage{geometry}
\newtheorem{theorem}{Theorem}

\newtheorem{fact}{Fact}
\newtheorem{example}{Example}
\newtheorem{lemma}{Lemma}
\newtheorem{proposition}{Proposition}
\newtheorem{remark}{Remark}

\newenvironment{proof}[1][Proof]{\emph{#1.} }{\  \hfill $\square $ \vspace{5 pt}}

\usepackage{natbib}
\usepackage{amssymb}
\usepackage{amsmath}

\usepackage{mathpazo}
\usepackage{enumerate}

\usepackage{amsfonts}

\usepackage{amssymb}

\usepackage{enumerate}

\usepackage{mathrsfs}

\usepackage{cancel}

\usepackage{pgf,tikz}

\usepackage{graphics}
\usepackage{color}

\usetikzlibrary{decorations.pathreplacing,angles,quotes}
\usetikzlibrary{arrows.meta}
\usetikzlibrary{arrows, decorations.markings}
\tikzset{myptr/.style={decoration={markings,mark=at position 1 with %
       {\arrow[scale=2,>=stealth]{>}}},postaction={decorate}}}

\usepackage{hyperref}
\hypersetup{
    colorlinks,
    citecolor=blue,
    filecolor=black,
    linkcolor=blue,
    urlcolor=blue
}

\usepackage{pgf,tikz}
\usetikzlibrary{arrows}

\usepackage{fancyhdr}

\usepackage{soul}

\usepackage{emptypage}

\newcommand*\samethanks[1][\value{footnote}]{\footnotemark[#1]}

\usepackage[T1]{fontenc}
\DeclareFontFamily{T1}{calligra}{}
\DeclareFontShape{T1}{calligra}{m}{n}{<->s*[1.44]callig15}{}
\DeclareMathAlphabet\mathcalligra   {T1}{calligra} {m} {n}

\usepackage{multirow}
\usepackage{array}
\usepackage{mathtools}
\usepackage{arydshln}
\usepackage{threeparttable}
\usepackage{booktabs,caption}

\usepackage{ifthen}
\newboolean{showcomments}
\setboolean{showcomments}{true}

\newcommand{\pablo}[1]{  \ifthenelse{\boolean{showcomments}}
{\textcolor{green!50!black}{(T: #1)}}{}}
\newcommand{\marcelo}[1]{\ifthenelse{\boolean{showcomments}}
{\textcolor{red}{(M: #1)}}{}}
\newcommand{\agustin}[1]{  \ifthenelse{\boolean{showcomments}}
{\textcolor{blue!50!black}{(T: #1)}}{}}

\usepackage{orcidlink}

\begin{document}

\title{Non-obvious manipulability in division problems with general preferences%
\thanks{%
We thank William Thomson for his detailed comments, as well as the participants of the Theory Seminar at the Department of Economics, University of Rochester, and the FOVI Seminar at the University of Chile.  
We acknowledge the financial support
from UNSL through grants 032016, 030120, and 030320, from Consejo Nacional
de Investigaciones Cient\'{\i}ficas y T\'{e}cnicas (CONICET) through grant
PIP 112-200801-00655, and from Agencia Nacional de Promoción Cient\'ifica y Tecnológica through grant PICT 2017-2355.}}


\author{R. Pablo Arribillaga\thanks{
Instituto de Matem\'{a}tica Aplicada San Luis (UNSL and CONICET) and Departamento de Matemática, Universidad Nacional de San
Luis, San Luis, Argentina. Emails: \href{mailto:rarribi@unsl.edu.ar}{rarribi@unsl.edu.ar} 
and \href{mailto:abonifacio@unsl.edu.ar}{abonifacio@unsl.edu.ar} 
} \hspace{1.5 pt} \orcidlink{0000-0002-0521-0301} \and Agustín G. Bonifacio\samethanks[2] 
\hspace{1.5 pt} \orcidlink{0000-0003-2239-8673}}

\date{\today}

\maketitle

\begin{abstract}

In problems involving the allocation of a single non-disposable commodity, we study rules defined on a general domain of preferences requiring only that each preference exhibit a unique global maximum. Our focus is on rules that satisfy a relaxed form of \emph{strategy-proofness}, known as \emph{non-obvious manipulability}. We show that the combination of \emph{efficiency} and \emph{non-obvious manipulability} leads to impossibility results, whereas  weakening \emph{efficiency} to \emph{unanimity} gives rise to a large family of well-behaved \emph{non-obviously manipulable} rules.



\bigskip

 \bigskip


\end{abstract}

\section{Introduction}

A division problem concerns the full allocation of a perfectly divisible commodity among a set of agents. A rule is defined as a systematic procedure that specifies, for every such problem, a feasible allocation. Since the seminal contribution of \cite{sprumont1991division}, the literature has been devoted to the identification of rules that simultaneously satisfy the three classical desiderata of mechanism design: efficiency, equity, and robustness to strategic manipulation. The bulk of this research has been conducted under the strong assumption of single-peaked preferences, whereby an agent’s welfare increases with consumption up to a critical level---the peak---and decreases beyond it \citep[see][]{thomson1994consistent,thomson1994resource,thomson1995population,barbera1997strategy}.\footnote{Other relevant domain of preferences that has been studied involves single-dipped preferences,  whereby an agent’s welfare decreases with consumption up to a critical level---the dip---and increases beyond it \citep[see][]{klaus1997strategy, gong2024mechanisms}.} 




From the perspective of robustness to manipulation, the property of ``strategy-proof\-ness''---according to which misreporting preferences is never better than truth-telling---has played a central role in the search for well-behaved rules. Yet strategy-proofness is generally a demanding requirement. While opportunities for manipulation may be abundant, agents often fail to recognize or exploit them due to limited information about others’ behavior or  cognitive constraints. Assuming that an agent knows all possible outcomes under any preference misreport and under truth-telling,  \cite{troyan2020obvious} propose a simple and tractable way to recognize a misreport as a manipulation by considering best/worst-case scenarios to formalize the notion of obvious manipulation. A manipulation is obvious if it either makes the agent better off than truth-telling in the worst case or makes the agent better off than truth-telling in the best case. An allotment rule is ``non-obviously manipulable'' if it has no obvious manipulation.

In the domain of single-peaked preferences, non-obviously manipulable rules are abundant. \cite{arribillaga2025not} characterize a broad class of such rules that, in addition to non-obvious manipulability, satisfy efficiency, own-peak-onliness, and the equal division guarantee. A rule is ``own-peak-only'' if the sole information used from an agent’s preference to determine their allotment is the agent’s peak, and it satisfies the ``equal division guarantee'' if an agent receives their peak allocation whenever that peak coincides with the equal division outcome. Notably, when the property of ``consistency''\footnote{Consistency says that the allocation recommended by a rule in a given problem should remain in agreement with the allocation it would recommend in a reduced problem obtained when some participants leave with their assigned allotments.} is also invoked, the  ``uniform'' rule is single out \citep{arribillaga2025obvious}. However, \cite{arribillaga2025not} also show that the domain of single-peaked preferences can only be slightly enlarged (to the domain of single-plateaued preferences) to still be able to define rules satisfying the four properties that characterize this family of rules.


\begin{figure}[t]
\centering
\begin{tikzpicture}[scale=0.9]
  \pgfmathdeclarefunction{Ri}{1}{%
    \pgfmathparse{
      3*exp(-((#1-6)^2)/5)              
      - 2*exp(-((#1-3)^2)/1)            
      + 0.6*exp(-((#1-1)^2)/0.30)       
      + 0.15/(1 + exp(-(#1-9)))         
    }%
  }

  \def\ax{5}

  \draw[->] (0,-2) -- (10.5,-2);
  \draw[->] (0,-2) -- (0,4);

  \draw[blue,thick, domain=0:10, smooth, samples=400]
    plot (\x, {Ri(2*(\ax)-\x)});

  \pgfmathsetmacro{\xT}{6}
  \pgfmathsetmacro{\xD}{2.85}
  \pgfmathsetmacro{\yT}{Ri(\xT)}
  \pgfmathsetmacro{\yD}{Ri(\xD)}

  \pgfmathsetmacro{\xTref}{2*(\ax)-\xT}
  \pgfmathsetmacro{\xDref}{2*(\ax)-\xD}

  \draw[dashed] (\xTref,-2) -- (\xTref,\yT);
  \fill (\xTref,-2) circle (2pt) node[below] {$t(R_i)$};


  \node at (6.5,1) {$R_i$};
\end{tikzpicture}
\caption{\textsf{Preference $R_i$ with a global maximum at $t(R_i)$.}}
\label{fig:spgd}
\end{figure}

Although single-peaked preferences constitute a significant domain, in many practical contexts agents exhibit preferences that possess a unique global maximum, but whose utility does not decrease monotonically on both sides of the optimum. Figure~\ref{fig:spgd} depicts a typical preference of the kind discussed. Such patterns naturally arise in settings with threshold effects, complementarities, or economies of scale, which are entirely excluded under the single-peaked assumption. 


To illustrate these phenomena, consider preferences with ``internal valleys''. For example, a group of farmers in need of irrigation water. Each farmer has a technologically optimal amount that maximizes crop yield. However, a subsidy for water use is forfeited whenever allocations fall within a ``forbidden band''. Hence, allocations in this interval are strictly worse than either having slightly less (and keeping the subsidy) or slightly more (and recovering economies of scale). Other situations of this type include fishing quotas with observer requirements (where quotas below a certain threshold are exempt from costly onboard observation) and electricity allocation with tariff notches (where intermediate consumption entails high costs without scale benefits).
As additional illustrations, notice that many allocation problems involve multidimensional decision packages (space, time, equipment, money). When projected onto a single dimension, preferences may lose their single-peakedness, yet retain a well-defined global maximum.\footnote{As an example, 
consider a research group requiring physical space. Utility depends both on office and lab space with an optimum at 40 $m^2$  
of office space and 20 $m^2$ of lab space. 
If we reduce this to total square meters, 
the global maximum appears at 60 $m^2$, but the projected function oscillates: 
with 70 $m^2$ partially unused, costly space provides lower utility; with 80 $m^2$ a second laboratory can be installed and some utility is recovered, but not exceeding the optimum.}   

 The previous examples motivate the search for a richer domain of preferences in which to study non-obviously manipulable rules. This is the goal of our work: the only domain restriction we impose is that each preference has a unique global maximum. This assumption is in line with recent work that recognizes satiation of preferences as an important topic ``that deserves more thorough treatment'' \citep[see][]{chatterji2025strategy}.

The results we obtained in our general domain of preferences can be summarized as follows. First, insisting on efficiency precipitates impossibility results. We show that no rule satisfies efficiency, strategy-proofness, and the ``equal division lower bound'', according to which no agent can be allotted an amount they consider worse than equal division (Theorem \ref{Impossibility: eff sp and edlb}). To deal with this impossibility, we then relax strategy-proofness to the combination of non-obvious manipulability and own-peak-onliness.\footnote{The properties of efficiency and strategy-proofness combined imply own-peak-onliness \citep[see][]{sprumont1991division,ching1994alternative}.} Unfortunately, adding efficiency leads to dictatorship in two-agent problems (Theorem \ref{Impossibility: eff NOM and OPO}). Aditionally, for more than two agents, efficient and own-peak-only rules are ``bossy'': an agent, by reporting a new preference, may affect the allotment of other agent without affecting their own (Theorem \ref{Impossibility: eff OPO and non-bossy}).
One might suspect that these negative results stem from the universality of our domain. However, all our impossibility proofs remain valid in much more restrictive domains that satisfy only a minimal richness condition: for every pair of alternatives, there exists a preference placing one at the top and the other at the bottom. For instance, the domain of continuous preferences in \cite{barbera1990strategy} satisfies this minimal richness condition.\footnote{These results are in line with similar findings obtained in classical exchange economies. \cite{momi2013note} shows that (i) efficient and strategy-proof rules do not guarantee positive consumption for all agents, and (ii) efficient, strategy-proof, and non-bossy rules are dictatorial. See also \cite{zhou1991inefficiency} and \cite{serizawa2002inefficiency}.}
We have chosen to work throughout the paper with our general domain because we believe it is of independent interest, and also because the positive results we obtain—and present below—are more compelling when established on a broader domain.


Next, in search of possibility results, we weaken efficiency to ``unanimity'', which requires that the rule assign each agent their peak whenever the sum of all peaks equals the social endowment. We show that the conjunction of unanimity, non-obvious manipulability, own-peak-onliness, and the equal division guarantee characterizes a large family of rules, which we term ``agreeable'' (Theorem \ref{characterization}). To define one such rule, assume that each agent is entitled to an equal share of the social endowment. Next, for each preference profile, select an (agreeable) coalition whose members’ peaks sum exactly to the portion of the endowment to which they are collectively entitled. Then, the rule assigns to each member of the coalition their peak and gives equal division to all remaining agents.
Regrettably, some agreeable rules can be bossy. To address this issue, we identify and characterize all agreeable rules that preclude such behavior (Proposition \ref{prop non-bossy}). Finally, by way of example, we present a natural family of agreeable and non-bossy rules (Proposition \ref{non-bossy agreeable example}).
In this family, each member is parameterized by a collection of nested coalitions, and for each preference profile the selected agreeable coalition is the maximal feasible one within the given collection.

Non-obviously manipulable rules have attracted considerable attention in recent years. In addition to our work on the division problem with single-peaked preferences, \cite{aziz2021obvious} and  \cite{arribillaga2024obvious} analyze non-peaks-only and peaks-only rules, respectively, in voting contexts. Related studies of obvious manipulations in other settings include  \cite{ortega2022obvious}, \cite{psomas2022fair}, and  \cite{arribillaga2025note,arribillaga2025obvious2}.



The remainder of the paper is organized as follows. Section~\ref{Preliminaries} introduces the model and the notion of non-obvious manipulability. Section~\ref{section negative} presents several impossibility results under efficiency. Section~\ref{section positive} consists of two parts. Subsection~\ref{subsection positive characterization} develops our main positive results by relaxing efficiency to unanimity and characterizing the class of agreeable rules, whereas Subsection~\ref{subsection non-bossy} characterizes the subclass of non-bossy agreeable rules and provides a broad family of such rules with a simple description. Section~\ref{further results} extends the analysis to economies with individual endowments and introduces peak-responsive rules. Finally, Section~\ref{final remarks} concludes.

\section{Preliminaries}\label{Preliminaries}


An amount $\Omega \in \mathbb{R}_{++}$ of an infinitely divisible commodity, the social endowment, has to be distributed among a set of agents $N=\{1,2,\ldots, n\}$ with $n\geq 2$. Each $i \in N$ is equipped with  a preference relation $R_i$ defined over $\mathbb{R}_+$. The only restriction on preferences is that each has a unique \textbf{peak}, denoted by $t(R_i)$. Call $P_i$ and $I_i$ to the strict preference and indifference relations associated with $R_i,$ respectively. Denote by $\mathcal{U}$ the domain of all such preferences. 
An  \textbf{economy} consists of a profile of preferences $R=(R_j)_{j \in N} \in \mathcal{U}^n$. The set of \textbf{(feasible) allotments} is $X=\{(x_1, \ldots, x_n) \in  \mathbb{R}^n_+ : \sum_{j\in N}x_j= \Omega\}$. 
An   \textbf{(allotment) rule} is a function $\varphi: \mathcal{U}^n \longrightarrow X$. 

Given a rule  $\varphi$, some desirable properties we consider are listed next. The first one entails the usual Pareto optimality criterion. Under this condition, for each economy, the allocation selected by the rule should be such that there is no other allocation that all agents find at least as desirable and at least one agent (strictly) prefers.

\vspace{5 pt}

\noindent \textbf{Efficiency:} For each $R\in \mathcal{U}^n$, there is no $x \in X$ such that $x_iR_i\varphi_i(R)$ for each $i \in N$ and $x_iP_i\varphi_i(R)$ for some $i \in N.$ 

\vspace{5 pt}

\noindent A weaker notion states that when the sum of the peaks is equal to the endowment, the rule must assign, for each agent, his peak amount. Formally,   

\vspace{5 pt}

\noindent \textbf{Unanimity:} For each $R\in \mathcal{U}^n$ such that $\sum_{i\in N}t(R_i)=\Omega$, we have $\varphi_i(R)=t(R_i)$ for each $i \in N.$

\vspace{5 pt}

The following is an informational simplicity property requiring that if an agent unilaterally changes his preference for another one with the same peak, then his allotment remains unchanged. Because of their simplicity, rules satisfying this property are important rules
in their own right and are both useful in practice and extensively studied in the literature.

\vspace{5 pt}

\noindent
\textbf{Own-peak-onliness:} For each  $R \in \mathcal{U}^n$ each $i \in N,$ and each $R_i' \in  \mathcal{U}$ such that $t(R_i')=t(R_i),$ we have $\varphi_i(R )=\varphi_i(R_i', R_{-i}).$

\vspace{5 pt}

Our first fairness property says that each agent must receive an amount he considers at least as good as equal division. 

\vspace{5 pt}
\noindent
\textbf{Equal division lower bound:} For each  $R \in \mathcal{U}^n$ and each $i \in N$, we have $\varphi_i(R)\ R_i \ \frac{\Omega}{n}$.

\vspace{5 pt}
A mild fairness requirement that weakens the \emph{equal division lower bound}, introduced in \cite{arribillaga2025not}, is presented next. It says that an agent should obtain his share of equal division whenever his peak is the equal division.

\vspace{5 pt}
\noindent
\textbf{Equal division guarantee:} For each  $R \in \mathcal{U}^n$ and each $i \in N$ such that $t(R_i)=\frac{\Omega}{n}$, we have $\varphi_i(R)=\frac{\Omega}{n}$.

\vspace{5 pt}


Finally, for our incentive compatibility property, we need some definitions. Let  $i\in N$ and $R_i\in \mathcal{U}$.
The \textbf{option set attainable with $\boldsymbol{R_{i}$ at $\varphi}$}  is 
\begin{equation*}
O^\varphi(R_{i})=\left\{x \in [0, \Omega] 
\ : \ x=\varphi_i(R_{i},R_{-i}) \text{ for some } R_{-i}\in \mathcal{U}^{n-1}\right\}.\footnote{\cite{barbera1990strategy} were the first to use option sets in the context of preference aggregation.}
\end{equation*}

Preference   $R_i' \in \mathcal{U}$ is an \textbf{obvious manipulation  of  
$\boldsymbol{\varphi$ at $R_i}$} if:
\begin{enumerate}[(i)]
    \item there is  $R_{-i} \in \mathcal{U}^{n-1}$ such that 
$\varphi_i(R_i', R_{-1}) \ P_i \ \varphi_i(R_i, R_{-i})$; and

\item for each 
$x' \in O^\varphi(R_{i}^{\prime })$ there is $x \in O^\varphi(R_{i})$ such that $x'P_i x.$

\end{enumerate}
When $R_i'$ satisfies (i) we say that $R_i'$ is a \textbf{manipulation of $\boldsymbol{\varphi$ at $R_i}$}. A manipulation becomes obvious if (ii) holds, i.e., if each possible outcome under the manipulation is strictly better than some possible outcome under truth-telling.


The following property precludes such behavior. 
\vspace{5 pt}

\noindent \textbf{Non-obvious manipulability:} For each  $i \in N$, and each  $R_i \in \mathcal{U}$ there is no obvious manipulation of $\varphi$ at $R_i$.  

\vspace{5 pt}

If a rule has no manipulation, it is \textbf{strategy-proof}. Clearly, \emph{strategy-proofness} implies \emph{non-obvious manipulability}.

\begin{remark}   \em  
The original definition of obvious manipulation introduced by \cite{troyan2020obvious} says that the best possible outcome in $O^\varphi(R_{i}^{\prime })$ should be better than the best possible outcome in $O^\varphi(R_{i})$ or the worst possible outcome in $O^\varphi(R_{i}^{\prime })$ should be better than the worst possible outcome in $O^\varphi(R_{i})$.  
However, under unanimity (or efficiency), the best possible outcome under truth-telling for an agent is always the agent's peak alternative in our model and thus no manipulation is obvious by considering such outcomes. Furthermore, when the set of alternatives is infinite (as in our case), the worst possible outcomes may not be well-defined and a more general definition is necessary. For a formal discussion and comparison of our definition of \emph{non-obvious manipulability} and the original presented by \cite{troyan2020obvious}, see \cite{arribillaga2025not}.
\end{remark}

\section{Negative results with efficiency} \label{section negative}


As a first goal, it would be valuable to identify (if any) rules that satisfy  \emph{efficiency, strategy-proofness}, and some minimal fairness requirement.  In economies with multiple goods and classical preferences, this question has been extensively studied since the seminal work of \cite{hurwicz1972informationally}, yielding predominantly negative results \citep[see][and references therein]{momi2013note}. 
In contrast, in economies with one perfectly divisible commodity and single-peaked preferences, there are good news: several ``sequential rules'' satisfy welfare lower bounds with respect to pre-specified individual endowments, that could be specialized to equal division \citep{barbera1997strategy}. Moreover, the uniform rule is the unique rule that, in adittion to being \emph{efficient} and \emph{strategy-proof}, satisfies either ``anonymity'' \citep{sprumont1991division}
 or  ``equal treatment of equals'', which says that agents with the same preference are required to obtain the same outcome \citep{ching1994alternative}. 

When extending the domain of preferences to our general domain, and in line with the Gibbard-Satterthwaite theorem, negative results are to be expected. 
Next, we show that the combination of \emph{efficiency} and \emph{strategy-proofness} is incompatible with the \emph{equal division lower bound}.\footnote{While we have searched for this result in the literature, we have not found it formally established. Therefore, we present it here for completeness.}


\begin{theorem}\label{Impossibility: eff sp and edlb}
No rule 
satisfies efficiency, strategy-proofness, and the equal division lower bound.  
\end{theorem}
\begin{proof} The proof is by induction on the number of agents. If $|N|=2$, we can fully describe the allotment problem as a social choice problem of a public good where we have to choose a single point in $[0,\Omega]$ that indicates the allotment of agent 1 and, therefore, also the allotment of agent 2. Given an allotment rule $\varphi: \mathcal{U}^2 \longrightarrow X$ consider the social choice rule $\overline{\varphi}:\mathcal{U}^2 \longrightarrow [0,\Omega]$ such that $\overline{\varphi}(R_1, R_2)=\varphi_1(R_1, \overline{R}_2)$ where $\overline{R}_2$ is such that $x\overline{R}_2 y$ if and only if $(\Omega-x) R_2 (\Omega-y)$ for each pair $x,y \in [0, \Omega]$. Assume that $\varphi$ satisfies \emph{efficiency} and \emph{strategy-proofness}. It is easy to see that $\overline{\varphi}$ satisfies those properties as well.\footnote{The adaptation of the definitions of \emph{efficiency} and \emph{strategy-proofness} to a social choice rule is straightforward.} 
By \cite{barbera1990strategy}, $\overline{\varphi}$ must be a dictatorship on the subdomain of continuous preferences contained in $\mathcal{U}$. Therefore, $\varphi$ does not satisfies the \emph{equal division lower bound} on such domain and, therefore, it does not satisfy this property on $\mathcal{U}$ either. 

Next, assume that $N=\{1,2,\ldots,n\}$ with $n\geq 3$ and, by the inductive hypothesis, that no rule with $n-1$ agents satisfies \emph{efficiency, strategy-proofness}, and the \emph{equal division lower bound}.  Assume, by the sake of contradiction,  that there is a rule $\varphi: \mathcal{U}^{n} \longrightarrow X$ satisfying all three properties. Let  $N'=N \setminus \{n\}$, $\Omega'=\Omega - \frac{\Omega}{n}$,  $X'=\{x \in  \mathbb{R}^{n-1}_+ : \sum_{j\in N'}x_j= \Omega'\}$ and define rule $\varphi': \mathcal{U}^{n-1} \longrightarrow X'$ as follows. For each $(R_1, \ldots, R_{n-1}) \in \mathcal{U}^{n-1}$ and each $j \in N \setminus \{n\}$, $$\varphi_j'(R_1, \ldots, R_{n-1})=\varphi_j(R_1, \ldots, R_{n-1},R_n')$$ where $R'_n \in \mathcal{U}$ is fixed and such that $t(R_n')=\frac{\Omega}{n}$. Note that $\varphi'$ is well-defined since, by the \emph{equal division lower bound}, for each $(R_1, \ldots, R_{n-1}) \in \mathcal{U}^{n-1}$, $\varphi_n(R_1, \ldots, R_{n-1},R_n')=\frac{\Omega}{n}$. As $\varphi$ is \emph{efficient} and \emph{strategy-proof}, so is $\varphi'$. Furthermore, $\frac{\Omega'}{|N'|}=\frac{\Omega-\frac{\Omega}{n}}{n-1}=\frac{\Omega}{n}$. Therefore, since $\varphi$ satisfies the \emph{equal division lower bound}, so does $\varphi'$. This contradicts our inductive hypothesis. Hence, $\varphi$ does not satisfy all three properties.   
\end{proof}

Given this first negative result, the conjunction of \emph{non-obvious manipulability} and \emph{own-peak-onliness} emerges as an interesting relaxation of \emph{strategy-proofness}. Remember that when preferences are single-peaked, \emph{own-peak-onliness} follows from \emph{efficiency} and \emph{strategy-proofness}
\citep[see, for example][]{sprumont1991division,ching1994alternative}.  
In order to find positive results, we also relax the \emph{equal division lower bound} to the minimal fairness condition of the \emph{equal division guarantee}.

However, adapting the ideas in the proof of Theorem 3 in \cite{arribillaga2025not} to our general domain, it can  be shown that  \emph{efficiency, non-obvious manipulability, own-peak-onliness}, and the \emph{equal division guarantee} are also incompatible in our setting. For completeness, we present this result next.

\begin{fact}{\citep{arribillaga2025not}}\label{fact}
No rule 
satisfies efficiency, non-obvious manipulability, own-peak-onliness, and the equal division guarantee. 
\end{fact}

Thus, if we want to maintain \emph{efficiency, non-obvious manipulability} and \emph{own-peak-onliness}, we need to relax the \emph{equal division guarantee}. Nevertheless, for two-agent economies, the combination of the first three properties implies a dictatorship.




\begin{theorem}\label{Impossibility: eff NOM and OPO} Let $|N|=2$. If $\varphi$ satisfies efficiency, non-obvious manipulability, and own-peak-onliness,  then $\varphi$ is dictatorial, i.e., there is $i \in N$ such that, for each $R \in \mathcal{U}^n$, $\varphi_i(R)=t(R_i)$.  
\end{theorem}
\begin{proof}
    Let $N=\{1,2\}$ and assume $\varphi$ satisfies \emph{efficiency, non-obvious manipulability} and \emph{own-peak-onliness}. The proof consists of three steps:

    \medskip

    \noindent \textbf{Step 1: for each $\boldsymbol{R \in \mathcal{U}^2$ there is $i \in \{1,2\}$ such that $\varphi_i(R)=t(R_i)}$.} Otherwise, there is $\overline{R} \in \mathcal{U}^2$ such that $\varphi_i(\overline{R})\neq t(R_i)$ for each $i \in \{1,2\}$. Since $\Omega-t(\overline{R}_1)\neq \varphi_2(\overline{R})$, there is $R_2' \in \mathcal{U}$ with $t(R'_2)=t(\overline{R}_2)$ and $(\Omega-t(\overline{R}_1)) \ P_2' \ \varphi_2(\overline{R})$. By \emph{own-peak-onliness}, $\varphi(\overline{R})=\varphi(\overline{R}_1, R_2')$. Let $x \equiv (t(\overline{R}_1), \Omega-t(\overline{R}_1))$. Then, $x_1 \ \overline{P}_1 \ \varphi_1(\overline{R}_1, R_2')$ and $x_2 \ P_2' \ \varphi_2(\overline{R}_1, R_2')$, contradicting the \emph{efficiency} of $\varphi$.  

    \medskip 

    \noindent \textbf{Step 2: for each $\boldsymbol{x \in [0,\Omega]$ there is $i \in \{1,2\}$ such that $\varphi_i(R)=t(R_i)}$ for each  $\boldsymbol{R \in \mathcal{U}^2}$ with $\boldsymbol{t(R_i)=x}$.} Let $x \in [0,\Omega]$ and consider $\overline{R} \in \mathcal{U}^2$ such that $t(\overline{R}_1)=t(\overline{R}_2)=x$. By Step 1 we can assume, w.l.o.g., that $\varphi_1(\overline{R})=t(\overline{R}_1)$. Next, we prove that 
    \begin{equation}\label{step 2 eq 1}
        \varphi_1(\overline{R}_1, R_2)=t(\overline{R}_1) \text{ for each }R_2 \in \mathcal{U}.
    \end{equation}
    Assume, by way of contradiction, that there is $R_2' \in \mathcal{U}$ such that \begin{equation}\label{step 2 eq 2}
        \varphi_1(\overline{R}_1, R_2') \neq t(\overline{R}_1).
    \end{equation}
    There are two cases to consider:
    \begin{itemize}
        \item[$\boldsymbol{1}$.] $\boldsymbol{x\neq \frac{\Omega}{2}}$. Notice first that  
        \begin{equation}\label{step 2 eq 3}
        \varphi_2(\overline{R}_1, R_2') \neq \varphi_2(\overline{R}).    
        \end{equation}
        Otherwise, $ \varphi_2(\overline{R}_1, R_2') = \varphi_2(\overline{R})=\Omega - t(\overline{R}_1)$ implies $\varphi_1(\overline{R}_1, R_2') = t(\overline{R}_1),$ contradicting \eqref{step 2 eq 2}. Also,  by this case's hypothesis  $t(\overline{R}_2) \neq \Omega - t(\overline{R}_1)=\varphi_2(\overline{R})$ and, thus,  $t(\overline{R}_2) \neq \varphi_2(\overline{R})$. Therefore, there is $\widetilde{P}_2 \in \mathcal{U}$ such that $t(\widetilde{P}_2)=t(\overline{P}_2)$ and $y \ \widetilde{P}_2 \ \varphi_2(\overline{R})$ for each $y \in [0,\Omega] \setminus \{\varphi_2(\overline{R})\}$. By \emph{own-peak-onliness}, $\varphi_2(\overline{R}_1, \widetilde{R}_2) = \varphi_2(\overline{R}).$ By \eqref{step 2 eq 3} and the definition of $\widetilde{P}_2$,  $$\varphi_2(\overline{R}_1, R_2') \ \widetilde{P}_2 \ \varphi_2(\overline{R}_1, \widetilde{R}_2)$$ and therefore $R_2'$ is a manipulation of $\varphi$ at $\widetilde{R}_2$. Next, let $R_1' \in \mathcal{U}$ be such that $t(R_1')\neq t(\overline{R}_1)$. By Step 1, $\varphi_2(R_1',R_2')=\Omega-t(R_1')$ or $\varphi_2(R_1',R_2')=t(R_2')$ (and $t(R_2') \neq \varphi_2(\overline{R})$ because otherwise by \eqref{step 2 eq 3} and Step 1 we contradict \eqref{step 2 eq 2}). Thus, in either case we obtain that $\varphi_2(R_1',R_2') \neq \varphi_2(\overline{R})=\varphi_2(\overline{R}_1, \widetilde{R}_2)$ and $R_2'$ is an obvious manipulation of $\varphi$ at $\widetilde{R}_2$, contradicting that $\varphi$ is \emph{non-obviously manipulable}. Hence, \eqref{step 2 eq 1} holds. 

        \item[$\boldsymbol{2}$.] $\boldsymbol{x= \frac{\Omega}{2}}$. Let $y \in [0, \Omega]\setminus \{x, \varphi_1(\overline{R}_1, R_2')\}$ and consider $\widehat{R}_1 \in \mathcal{U}$ such that $t(\widehat{R}_1)=y$. By the previous case, $\varphi_1(\widehat{R}_1, R_2')=t(\widehat{R}_1)=y$. Let $\widetilde{P}_1 \in \mathcal{U}$ be such that $t(\widetilde{P}_1)=t(\overline{P}_1)$ and $z \ \widetilde{P}_1 \ \varphi_1(\overline{R}_1, R_2')$ for each $z \in [0,\Omega] \setminus \{\varphi_1(\overline{R}_1, R_2')\}$. By \emph{own-peak-onliness}, $\varphi_1(\widetilde{R}_1, R_2')=\varphi_1(\overline{R}_1, R_2')$. By the definition of $\widetilde{P}_1$, $$\varphi_1(\widehat{R}_1, R_2') \ \widetilde{P}_1 \ \varphi_1(\widetilde{R}_1, R_2')$$ and therefore $\widehat{R}_1$ is a manipulation of $\varphi$ at $\widetilde{R}_1$. Next, let $\widehat{R}_2 \in \mathcal{U}$. By Step 1, $\varphi_1(\widehat{R}_1,\widehat{R}_2)=t(\widehat{R}_1)$ or  $\varphi_1(\widehat{R}_1,\widehat{R}_2)=\Omega -t(\widehat{R}_2)$. In either case, $\varphi_1(\widehat{R}_1,\widehat{R}_2)\neq \varphi_1(\overline{R}_1, R_2')$, so $\widehat{R}_1$ is an obvious manipulation of $\varphi$ at $\widetilde{R}_1$, contradicting that $\varphi$ is \emph{non-obviously manipulable}. Hence, \eqref{step 2 eq 1} holds.
    \end{itemize}
    To finish this step, let $R_1 \in \mathcal{U}$ be such that $t(R_1)=t(\overline{R}_1)$. The result follows from \emph{own-peak-onliness}. 

    \medskip 

    \noindent \textbf{Step 3: there is $\boldsymbol{i \in \{1,2\}}$ such that for each $\boldsymbol{x \in [0,\Omega]}$ and each $\boldsymbol{R \in \mathcal{U}^2}$ with $\boldsymbol{t(R_i)=x}$, $\boldsymbol{\varphi_i(R)=t(R_i)}$}. Assume this is not the case. By Step 2, there are distinct $x,y \in [0,\Omega]$ such that, w.l.o.g., 
    \begin{equation}\label{step 3 eq 1}
        \varphi_1(R)=p(R_1) \text{ whenever }R \in \mathcal{U}^2 \text{ is such that }t(R_1)=x,
    \end{equation}
    and 
    \begin{equation}\label{step 3 eq 2}
        \varphi_2(R)=p(R_2) \text{ whenever }R \in \mathcal{U}^2 \text{ is such that }t(R_2)=y.
    \end{equation}
    Let $\overline{R} \in \mathcal{U}^2$ be such that $t(\overline{R}_1)=x$ and $t(\overline{R}_2)=y$. Therefore, \eqref{step 3 eq 1}, \eqref{step 3 eq 2}, and  feasibility imply
    \begin{equation}\label{step 3 eq 3}
\Omega=\varphi_1(\overline{R})+\varphi_2(\overline{R})=x+y.
    \end{equation}
    Now, let $z \in [0,\Omega]\setminus \{x,y\}$. By Step 2 we can assume, w.l.o.g., that 
    \begin{equation}\label{step 3 eq 4}
        \varphi_1(R)=p(R_1) \text{ whenever }R \in \mathcal{U}^2 \text{ is such that }t(R_1)=z.
    \end{equation} Let $\widetilde{R}_1 \in \mathcal{U}$ be such that $t(\widetilde{R}_1)=z$. Using  \eqref{step 3 eq 4}, \eqref{step 3 eq 2}, and feasibility we get $$\Omega=\varphi_1(\widetilde{R}_1, \overline{R}_2)+\varphi_2(\widetilde{R}_1, \overline{R}_2)=z+y.$$ This last fact together with \eqref{step 3 eq 3} imply $z=x$, a contradiction. This completes the proof of this step.

\medskip

To conclude the proof of the Theorem, by Step 3, there is $i \in \{1,2\}$ such that $\varphi_i(R)=t(R_i)$ for each $R \in \mathcal{U}^2$, and thus $\varphi$ is \emph{dictatorial}. 
\end{proof}


For more than two agents, 
it is easy to find rules such that the change in an agent's preference may not alter the amount assigned to that agent but instead affect someone else's. 
The following property precludes such behavior:

\vspace{5 pt}

\noindent \textbf{Non-bossiness:} For each $R\in \mathcal{U}^n$, each $i\in N$ and $R_i'\in \mathcal{U}$,  $\varphi_i(R'_i,R_{-i})=\varphi_i(R)$ implies $\varphi(R'_i,R_{-i})=\varphi(R)$.

\vspace{5 pt}

If a rule violates this property, we say that it is \textbf{bossy}.

Our next result states that the combination of  \emph{efficiency} and \emph{own-peak-onliness} alone generates an incompatibility with  \emph{non-bossiness}.

\begin{theorem}\label{Impossibility: eff OPO and non-bossy} Let $|N|>2$. Then, no rule satisfies efficiency, own-peak-onliness, and non-bossiness.
\end{theorem}
\begin{proof}
    Assume by contradiction that there exits $\varphi$ satisfying \emph{efficiency}, and \emph{own-peak-onliness}, and \emph{non-bossiness}. Let us consider a preference profile $R\in \mathcal{U}^n$ where $t(R_i)=\Omega$ for each $i\in N$. As $|N|>2$, there are $j,j'\in N$ such that $\varphi_j(R)<\Omega$ and $\varphi_{j'}(R)<\Omega$. 
    Take $x,y\in [0, \Omega]\setminus\{\varphi_j(R),\varphi_{j'}(R)\}$ such that $x+y=\varphi_j(R)+\varphi_{j'}(R)$. By \emph{own-peak-onliness} and \emph{non-bossiness}, we can assume that $xP_j\varphi_j(R)$ and $yP_{j'}\varphi(R)$ and so $\varphi(R)$ is not \emph{efficient} because the allotment that assigns $x$ to $j$, $y$ to $j'$ and $\varphi_i(R)$ for each $i\in N\setminus \{j,j'\}$ is a Pareto improvement over $\varphi(R)$.  
\end{proof}


\begin{remark}\rm
    It is important to highlight that all our impossibility results  remain valid working in any subdomain $\mathcal{D} \subseteq \mathcal{U}$ with the following minimal richness condition: for each pair of distinct $x,y \in [0, \Omega]$ there is a preference $R_i \in \mathcal{D}$ such that $xP_izP_iy$ for each $z \in [0, \Omega]\setminus \{x,y\}$, i.e., $x$ is the peak and $y$ is the dip of preference $R_i$.

\end{remark}

\section{Positive results relaxing efficiency to unanimity: agreeable rules} \label{section positive}

In this section, we relax \emph{efficiency} to \emph{unanimity}. In this case, we get positive results and characterize the class of of all rules that satisfy \emph{unanimity, non-obvious manipulability, own-peak-onliness}, and the \emph{equal division guarantee}. Within such class we identify a wide subclass of \emph{non-bossy} rules.

\subsection{Agreeable rules: definition and characterization}\label{subsection positive characterization}

To define one such rule, assume that each agent is entitled to an equal share of the social endowment. Next, for each preference profile, select a coalition whose members’ peaks sum exactly to the portion of the endowment to which they are collectively entitled. Then, the rule assigns to each member of the coalition their peak and gives equal division to all remaining agents. 
Formally, given an economy $R \in \mathcal{U}^n$ and a coalition $S \subseteq N$, we say that \textbf{$\boldsymbol{S$ is agreeable for  $R}$} if $$\sum_{j \in S}t(R_j)=\frac{|S|}{|N|}\Omega.$$ For each $R \in \mathcal{U}^n$, denote by $\mathcal{A}(R)$ to the set of all agreeable coalitions for $R$.\footnote{Observe that $\emptyset \in \mathcal{A}(R)$ for each $R\in \mathcal{U}^n$.} An \textbf{agreeable selection} is a function $\mathcal{S}: \mathcal{U}^n \longrightarrow 2^N$ that satisfies:
\begin{enumerate}[(i)]
    \item $\mathcal{S}(R) \in \mathcal{A}(R)$ for each $R \in \mathcal{U}^n$,


    \item for each  $R \in \mathcal{U}^n$, each $i \in N$, and each $R_i' \in \mathcal{U}$ with $t(R_i')=t(R_i)$, $$i \in \mathcal{S}(R) \text{ if and only if } i \in \mathcal{S}(R_i',R_{-i}),$$

    \item  $N \in \mathcal{A}(R)$ implies $\mathcal{S}(R)=N$ for each $R \in \mathcal{U}^n$.
\end{enumerate}

An agreeable selection picks agreeable coalitions in an \emph{own-peak-only} way (conditions (i) and (ii)) and always chooses the grand coalition when such a coalition is agreeable (conditions (iii))  which will be necessary to unanimity .

We can associate a rule with each agreeable selection $\mathcal{S}$ as follows:

\medskip 

\noindent \textbf{Agreeable rule, $\boldsymbol{\varphi^{\mathcal{S}}}$:} for each $R \in  \mathcal{U}^n$ and each $i \in N$,  
$$\varphi_i(R)=\begin{cases}
    t(R_i) & \text{if } i \in \mathcal{S}(R) \\
    \frac{\Omega}{n} & \text{otherwise}
\end{cases}$$
where $\mathcal{S}:\mathcal{U}^n \longrightarrow 2^N$ is an agreeable selection.

\medskip

We can describe an agreeable rule by means of a two-step procedure. Fix a profile of preferences. In the first step, an agreeable coalition is selected. In the second step, agents in the agreeable coalition are given their peaks, whereas the remaining agents obtain equal division .   

Observe that these rules satisfy the following property: whenever either all agents  demand less than equal division or all agents demand more than equal division, everyone gets equal division. 

Next, we present an example of agreeable rule. 


\begin{example}
Consider the agreeable selection $\mathcal{S}_q: \mathcal{U}^n \longrightarrow 2^N$ such that $\mathcal{S}(R)=\emptyset$ if $N \notin \mathcal{A}(R)$. 
The agreeable rule $\varphi^{S_{q}}$ associated with this selection is called the \textbf{egalitarian status quo} rule. This rule assigns equal division to all agents except when we have unanimity (i.e., the sum of all peaks is equal to the social endowment), in which case each one receives their peak.    \hfill $\Diamond$ 
\end{example}

\begin{remark}\label{EDG} \em
    It is clear from the definition that agreeable rules satisfy the \emph{equal division guarantee}. 
\end{remark}

It turns out that agreeable rules meet our incentive compatibility criterion.

\begin{proposition}\label{t1}
    Any agreeable rule satisfies non-obvious manipulability.
\end{proposition}
\begin{proof}
Let $\varphi: \mathcal{U}^n \longrightarrow X$ be an agreeable rule. Then, there is an agreeable selection $\mathcal{S}: \mathcal{U}^n \longrightarrow 2^N$ such that $\varphi=\varphi^{\mathcal{S}}$. Let $i \in N$ and $R_i \in \mathcal{U}$. First, we show that 
\begin{equation}\label{prueba NOM 1}
        O^\varphi(R_i)=\begin{cases}
            \left\{t(R_i),\frac{\Omega}{n}\right\} & \text{ if }t(P_i)\leq \Omega\\
            \left\{\frac{\Omega}{n}\right\} & \text{ otherwise.}
        \end{cases}
    \end{equation}
If $t(R_i)>\Omega$, $i \notin \mathcal{S}(R)$.  Therefore, $O^\varphi(R_i)=\{\frac{\Omega}{n}\}$. Now, consider the case $t(P_i)\leq \Omega$. By the definition of an agreeable rule, $O^\varphi(R_i)\subseteq\left\{t(R_i), \frac{\Omega}{n}\right\}$. To see the reverse inclusion, 
let $R_{-i}\in \mathcal{U}^{n-1}$ be such that $\sum_{j\in N\setminus\{i\}}t(R_j)=\Omega-t(R_i)$. Then $N\in \mathcal{A}(R)$ and, by condition (iii) of the definition of agreeable rule, we have that $\mathcal{S}(R)=N$. Thus, $\varphi_i(R)=t(R_i)$ and  $t(R_i)\in O^\varphi(R_i)$. Furthermore, let $R_{-i}' \in \mathcal{U}^{n-1}$ be such that, for each $j\in N\setminus\{i\}$, we have $t(R_j')=\frac{\Omega}{n}$. By Remark \ref{EDG}, for each $j\in N\setminus\{i\}$, we have $\varphi_j(R_i,R_{-i}')=\frac{\Omega}{n}$. Then, by feasibility, $\varphi_i(R_i,R_{-i})=\frac{\Omega}{n}$ and thus $\frac{\Omega}{n}\in O^\varphi(R_i)$. This proves \eqref{prueba NOM 1}.

Next, let $R_i' \in \mathcal{U}$ be a manipulation of $\varphi$ at $R_i$. By \eqref{prueba NOM 1}, $\frac{\Omega}{n} \in O^\varphi(R_i')$ and there is no $x\in O^\varphi(R_i)$ such that $\frac{\Omega}{n} \mathrel{P_i} x$. Thus, $R_i'$ is not an obvious manipulation and hence $\varphi$ is \emph{non-obviously manipulable}.   
\end{proof}

Our main result, presented next, charaterizes the class of all agreeable rules. 

\begin{theorem}\label{characterization} A rule satisfies unanimity, non-obvious manipulability, own-peak-onliness, and the equal division guarantee if and only if it is an agreeable rule.
    
\end{theorem}
\begin{proof}
    ($\Longleftarrow$) By Proposition \ref{t1}, an agreeable rule satisfies \emph{non-obvious manipulability}. It is straightforward to see that an agreeable rule satisfies the other three properties as well.   

\noindent ($\Longrightarrow$) Assume that a rule $\varphi$ safisfies \emph{unanimity, non-obvious manipulability, own-peak-onliness}, and the \emph{equal division guarantee}. 

\noindent \textbf{Claim: for each  $\boldsymbol{R \in \mathcal{U}^n}$ and each $\boldsymbol{i \in N}$,  $\boldsymbol{O^\varphi(R_i)=}\begin{cases}
            \boldsymbol{\left\{t(R_i),\frac{\Omega}{n}\right\}} & \text{ if }\boldsymbol{t(P_i)\leq \Omega}\\
            \boldsymbol{\left\{\frac{\Omega}{n}\right\}} & \text{ otherwise.}
        \end{cases}$}

\noindent To proof the claim, it is clear by the \emph{equal division guarantee}  that $\frac{\Omega}{n}\in O^\varphi(R_i)$. Furthermore, by \emph{unanimity} and feasibility,  $t(R_i)\in O^\varphi(R_i)$ if and only if $t(R_i) \leq \Omega$.  Now, assume by contradiction that there is $x\in O^\varphi(R_i)\setminus \{\frac{\Omega}{n},t(R_i)\}.$ Then, there is $R_{-i}\in \mathcal{U}^{n-1}$ such that $\varphi_i(R_i,R_{-i})=x$. By \emph{own-peak-onliness}, we can assume that $\frac{\Omega}{n}\mathrel{P_i}x$. Let $R'_i \in \mathcal{U}$ be such that $t(R'_i)=\frac{\Omega}{n}$. By the \emph{equal division guarantee}, $\varphi_i(R'_i,R'_{-i})=\frac{\Omega}{n}$ for each $R'_{-i}\in \mathcal{U}^{n-1}$ and so $R'_i$ is an obvious manipulation of $\varphi$ at $R_i.$ This proves the Claim.

Next, define $\mathcal{S}^\star: \mathcal{U}^n \longrightarrow 2^N$ as follows. For each $R \in \mathcal{U}^n$, $\mathcal{S}^\star(R)=\{i\in N : \varphi_i(R)=t(R_i) \}$. Function $\mathcal{S}^\star$ is an agreeable selection. To see this, let $R \in \mathcal{U}^n$. By the Claim, the definition of $\mathcal{S}^\star$, and  feasibility,  we have that 

$$\sum_{i \in \mathcal{S}^\star(R)}t(R_i)=\frac{|\mathcal{S}^\star(R)|}{|N|}\Omega,$$ thus $\mathcal{S}^\star(R) \in \mathcal{A}(R)$ and condition (i) in the definition of agreeable selection is satisfied. Furthermore, the fact that $\mathcal{S}^\star(R)$ satisfies conditions (ii) and (iii) in the definition of agreeable selection follows from $\varphi$ being \emph{own-peak-only} and \emph{unanimous}, respectively.





Finally, from the definition of $\mathcal{S}^\star$ we have that $\varphi=\varphi^{\mathcal{S}^\star}$, and therefore $\varphi$ is an agreeable rule.
\end{proof}

Next, we present some examples that prove the independence of the axioms in Theorem \ref{characterization}.

\begin{itemize}
    \item Let $\widetilde{\varphi}:\mathcal{U}^n \longrightarrow X$ be such that, for each $R\in \mathcal{U}^n$ and each $i\in N$, $\varphi_i(R)=\frac{\Omega}{n}$ . Then $\widetilde{\varphi}$ satisfies all properties but \emph{unanimity}. 

\item Let $\varphi^\star:\mathcal{U}^n \longrightarrow X$ be such that, for each $R\in \mathcal{U}^n$ and each $i\in N$, 
$$\varphi^\star(R)=\left\{\begin{array}{l l }
(t(R_1), t(R_2), \ldots,t(R_n)) & \text{if } \sum_{i\in N}t(R_i)=\Omega  \\ 
\\
(\Omega, 0, \ldots 0) & \text{otherwise}\\
\end{array}\right.$$

Observe that this rule does not meet 
the \emph{equal division guarantee}. It is clear that $\varphi^\star$ is \emph{own-peak-only} and \emph{unanimous}. Now, we prove that $\varphi^\star$ also satisfies \emph{non-obvious manipulability}. Let  $i\in N$ and  $R_i\in \mathcal{U}$. Notice that $O^{\varphi^\star}(R_i)=\{\Omega,t(R_i)\}$ if $i=1$ and $O^{\varphi^\star}(R_i)=\{0,t(R_i)\}$ if $i\neq1$. Now, the fact that $\varphi^\star$ does not have obvious manipulations 
follows a similar argument as in Proposition \ref{t1}. Then, $\varphi^\star$ satisfies all properties but the \emph{equal division guarantee}.

\item \textbf{Uniform rule, $\boldsymbol{u}$:} For each $R\in \mathcal{U}^n,$ and each $i \in N,$
$$u_i(R)=\left\{\begin{array}{l l }
\min\{t(R_i), \lambda\} & \text{if } \sum_{j\in N}t(R_j)\geq \Omega\\
\max\{t(R_i), \lambda\} & \text{if } \sum_{j\in N}t(R_j)<\Omega\\
\end{array}\right.$$
where $\lambda \geq 0$ and solves $\sum_{j \in N}u_j(R)=\Omega.$

Observe that the uniform rule satisfies \emph{unanimity, own-peak-onliness}, and the \emph{equal division guarantee}. However, it does not satisfy \emph{non-obvious manipulability} because it is not an agreeable rule. To see this last observation consider $R \in \mathcal{U}^n$ such that $t(R_1)=0$ and $t(R_j)=\Omega$ for each $j \in N\setminus \{1\}$. Note that $u_1(R)=0$ but $\varphi_1(R)=\frac{\Omega}{n}$ for each agreeable rule $\varphi$.

\item  Let $\widehat{\varphi}:\mathcal{U}^n \longrightarrow X$ be such that, for each $R\in \mathcal{U}^n$ and each $i\in N$,  
$$\widehat{\varphi}(R)=\left\{\begin{array}{l l }
(t(R_1), t(R_2), \ldots,t(R_n)) & \text{if } \sum_{i\in N}t(R_i)=\Omega 
\\ 
(0, t(R_2), \ldots,t(R_n)) & \text{if } \sum_{i\in N\setminus\{1\}}t(R_i)=\Omega \text{ and } 0\mathrel{P_1}\frac{\Omega}{n}
\\
(\frac{\Omega}{n}, \frac{\Omega}{n}, \ldots, \frac{\Omega}{n}) & \text{otherwise}\\
\end{array}\right.$$

Observe that $\widehat{\varphi}$ is not an agreeable rule because it is not \emph{own-peak-only}. It is clear that $\widehat{\varphi}$ satisfies \emph{unanimity} and the \emph{equal division guarantee}. Now, we prove that $\widehat{\varphi}$ satisfies \emph{non-obvious manipulability} as well. For each $i\in N\setminus\{1\}$ and each $R_i\in \mathcal{U}$,  $O^{\widehat{\varphi}}(R_i)=\{\frac{\Omega}{n}, t(R_i)\}$. Then, the fact that $\widehat{\varphi}$ does not have obvious manipulations for agents in $N\setminus\{1\}$ follows the arguments used in the proof of  Proposition \ref{t1}. Let us next consider agent $1$. First, let $R_1\in \mathcal{U}$ be  such that $\frac{\Omega}{n}\mathrel{R_1}0$. Then, $O^{\widehat{\varphi}}(R_1)=\{\frac{\Omega}{n}, t(R_1)\}$.  Again, the fact that $\widehat{\varphi}$ does not have obvious manipulations for agent $1$ at $R_1$ follows the arguments used in the proof of  Proposition \ref{t1} together with the fact that  $\frac{\Omega}{n}\in O^\varphi(R'_1)$ for each $R'_1 \in \mathcal{U}$. Second, let $R_1\in \mathcal{U}$ be such that $0\mathrel{P_1}\frac{\Omega}{n}$. Then, $O^{\widehat{\varphi}}(R_1)=\{\frac{\Omega}{n}, 0,t(R_i)\}$. Furthermore, as $\frac{\Omega}{n}\in O^\varphi(R'_1)$ for each $R'_1 \in \mathcal{U}$ and $0\mathrel{P_1}\frac{\Omega}{n}$, agent $1$ does not have  obvious manipulations of $\widehat{\varphi}$ at $R_1.$

\end{itemize}
    

\subsection{Non-bossy agreeable rules}\label{subsection non-bossy}

Although agreeable rules offer an appealing and tractable description of rules satisfying \emph{unanimity, non-obvious manipulability, own-peak-onliness}, and the \emph{equal division guarantee}, they are not without shortcomings. In particular, some of them may display bossy behavior, whereby an agent can affect others’ allotments without altering their own. We therefore proceed in two steps. First, we characterize the subclass of agreeable rules that are \emph{non-bossy}, which includes as an example the egalitarian status quo rule. Second, we present a broad and natural family of \emph{non-bossy} agreeable rules that admit a simple and intuitive formulation.




To avoid technical details, given an  agreeable selection $S:\mathcal{U}^n \longrightarrow 2^N$ for each $R \in  \mathcal{U}^n$ we need to identify the agents in $\mathcal{S}(R)$ whose peaks are different from $\frac{\Omega}{n}$, i.e., we will work with the set     
$$\mathcal{S}^\circ(R)=\left\{i\in \mathcal{S}(R) \ : \ t(R_i)\neq\frac{\Omega}{n}\right\}.$$
Observe that $\mathcal{S}^\circ$ is sufficient to define $\varphi^{\mathcal{S}}$, so sometimes we will say that $\mathcal{S}^\circ$ is the \textbf{refined agreeable selection} associated to $\varphi^{\mathcal{S}}$.

\begin{remark}\label{peak} \em
         If an agreeable rule $\varphi^{\mathcal{S}}$ is \emph{non-bossy}, then both $\varphi^\mathcal{S}$ and its agreeable selection $\mathcal{S}^\circ$ are \emph{peaks-only}, i.e., for each pair $R,R' \in  \mathcal{U}^n$ such that $t(R_i)=t(R'_i)$ for each $i \in N$,  
\begin{equation}\label{peaks -only}
    \varphi^{\mathcal{S}}(R)=\varphi^{\mathcal{S}}(R') \text{ and } \mathcal{S}^\circ(R)=\mathcal{S}^\circ(R').
\end{equation}
\end{remark}

It is easy to see that the converse of the previous remark does not hold and thus the \emph{peaks-only} property does not imply \emph{non-bossiness}. The next result adds an extra condition to \emph{peaks-onliness} to make it equivalent to \emph{non-bossiness} for agreeable rules.

\begin{proposition} \label{prop non-bossy}
    An agreeable rule $\varphi^{\mathcal{S}}$ is non-bossy if and only if its associated agreeable selection $\mathcal{S}^\circ$ is peaks-only and is such that   
\begin{equation}\label{non-bossy2}
 \text{ for each } R \in  \mathcal{U}^n, i \in N, \text{ and } R'_i \in  \mathcal{U}, \text{ if }    i\notin \mathcal{S}^\circ(R) \cup  \mathcal{S}^\circ(R'_i,R_{-i}) \text{ then }
\mathcal{S}^\circ(R)=\mathcal{S}^\circ(R'_i,R_{-i}).
\end{equation}
\end{proposition}
\begin{proof}
    ($\Longleftarrow$) Assume that 
    $\mathcal{S}^\circ$ is \emph{peaks-only} and \eqref{non-bossy2} holds. We will prove that $\varphi^{\mathcal{S}}$ is \emph{non-bossy}. Let $R \in  \mathcal{U}^n, i \in N, \text{ and } R'_i \in  \mathcal{U}$. Assume that $\varphi^{\mathcal{S}}_i(R)=\varphi^{\mathcal{S}}_i(R'_i,R_{-i})$. There are two cases.

    \noindent \textbf{Case 1:}  $\boldsymbol{\varphi^{\mathcal{S}}_i(R)\neq\frac{\Omega}{n}}$. Then, $\varphi^{\mathcal{S}}_i(R)=t(R_i)\neq\frac{\Omega}{n}$. As $\varphi^{\mathcal{S}}_i(R)=\varphi^{\mathcal{S}}_i(R'_i,R_{-i})$, we have that $t(R'_i)=t(R_i)\neq\frac{\Omega}{n}$. Then, by \emph{peaks-onliness},
    $$\mathcal{S}^\circ(R)=\mathcal{S}^\circ(R'_i,R_{-i}),$$
    which implies that, for each $j\in N\setminus\{i\},$
    $$\varphi^{\mathcal{S}}_j(R)=\varphi^{\mathcal{S}}_j(R'_i,R_{-i}).$$

    \noindent \textbf{Case 2:} $\boldsymbol{\varphi^{\mathcal{S}}_i(R_i,R_{-i})=\frac{\Omega}{n}}$. If $t(R'_i)=t(R_i)$ the proof follows from \emph{peaks-onliness} as in the previous case, so assume that $t(R'_i)\neq t(R_i)$. There are two subcases:  
    
     \noindent \textbf{Subcase 2.a:} $\boldsymbol{t(R_i)=\frac{\Omega}{n}\neq t(R'_i)}$ (the case  $t(R'_i)=\frac{\Omega}{n}\neq t(R_i)$ is symmetric). Then, as $\frac{\Omega}{n}=\varphi^{\mathcal{S}}_i(R)=\varphi^{\mathcal{S}}_i(R'_i,R_{-i})$ by definition of $\varphi^{\mathcal{S}}$,
    $$i\notin \mathcal{S}^\circ(R'_i,R_{-i})$$ 
    Furthermore, as $t(R_i)=\frac{\Omega}{n}$, $i\notin \mathcal{S}^\circ(R)$.  
    Therefore, as \eqref{non-bossy2} holds, 
    $$\mathcal{S}^\circ(R)=\mathcal{S}^\circ(R'_i,R_{-i})$$
    which implies that 
    $$\varphi^{\mathcal{S}}_j(R_i,R_{-i})=\varphi^{\mathcal{S}}_j(R'_i,R_{-i})$$ 
    for each $j\in N\setminus\{i\}.$

    \noindent \textbf{Subcase 2.b: $\boldsymbol{t(R_i)\neq\frac{\Omega}{n}$ and $t(R'_i)\neq\frac{\Omega}{n}}$}. Then, as $\frac{\Omega}{n}=\varphi^{\mathcal{S}}_i(R)=\varphi^{\mathcal{S}}_i(R'_i,R_{-i})$ by definition of $\varphi^{\mathcal{S}}$,
    $$i\notin \mathcal{S}^\circ(R) \text{ and } i\notin \mathcal{S}^\circ(R'_i,R_{-i})$$ 
    Therefore, as \eqref{non-bossy2} holds, 
    $$\mathcal{S}^\circ(R)=\mathcal{S}^\circ(R'_i,R_{-i})$$
    which implies that  
    $$\varphi^{\mathcal{S}}_j(R_i,R_{-i})=\varphi^{\mathcal{S}}_j(R'_i,R_{-i})$$ 
    for each $j\in N\setminus\{i\}.$
    
     \noindent ($\Longrightarrow$) Assume that $\varphi^{\mathcal{S}}$ is \emph{non-bossy}. By Remark \ref{peak}, $\mathcal{S}^\circ$ is \emph{peaks-only}. Suppose  by contradiction that \eqref{non-bossy2} does not hold. Then,  there are $R \in  \mathcal{U}^n$, $i \in N$, and $R'_i \in  \mathcal{U}$ such that
    \begin{equation*}
    i\notin \mathcal{S}^\circ(R) \cup \mathcal{S}^\circ(R'_i,R_{-i}) \text{ and } \mathcal{S}^\circ(R)\neq \mathcal{S}^\circ(R'_i,R_{-i}).
\end{equation*}
    Therefore, by definition of $\varphi^{\mathcal{S}}$, 
    \begin{equation}\label{UNO}
    \varphi^{\mathcal{S}}_i(R)=\frac{\Omega}{n}=\varphi^{\mathcal{S}}_i(R'_i,R_{-i}).    
    \end{equation}
     Furthermore, as $\mathcal{S}^\circ(R)\neq \mathcal{S}^\circ(R'_i,R_{-i})$, it follows that  $\mathcal{S}^\circ(R)\setminus \mathcal{S}^\circ(R'_i,R_{-i})\neq\emptyset$ or $\mathcal{S}^\circ(R'_i,R_{-i})\setminus \mathcal{S}^\circ(R) \neq\emptyset$. Assume $\mathcal{S}^\circ(R)\setminus \mathcal{S}^\circ(R'_i,R_{-i})\neq\emptyset$ (the other case is symmetric). Let $j\in \mathcal{S}^\circ(R)\setminus \mathcal{S}^\circ(R'_i,R_{-i})$. By definition of $\varphi^{\mathcal{S}}$ and $\mathcal{S}^\circ$ we have that 
     \begin{equation}\label{DOS}
         \varphi^{\mathcal{S}}_j(R_i,R_{-i})=t(R_j)\neq \frac{\Omega}{n}=\varphi^{\mathcal{S}}_j(R'_i,R_{-i})
     \end{equation}
     Thus, by \eqref{UNO} and \eqref{DOS}, we obtain  contradiction to \emph{non-bossiness}. Hence, \eqref{non-bossy2} holds.
\end{proof}

Next, we describe a wide and natural family of non-bossy, agreeable rules with a straightforward and intuitive formulation.
Let $\mathcal{C}=\{C_i\}_{i=1}^T$ be a \textbf{collection of nested coalitions} such that $$N=C_1 \supsetneq C_2 \supsetneq \ldots \supsetneq C_T=\emptyset.$$ 
Given one such collection $\mathcal{C}=\{C_i\}_{i=1}^T$, define its associated agreeable  selection $\mathcal{S}^\mathcal{C}$ as follows. For each $R \in \mathcal{U}^n$, $$\mathcal{S}^\mathcal{C}(R)=C_t \text{ whenever }C_t \in \mathcal{A}(R) \text{ and }C_{t'} \notin \mathcal{A}(R) \text{ for each }t'<t.$$ 

For example, the egalitarian status quo rule is induced by the collection of nested coalitions $\mathcal{C}=\{C_1=N, C_2=\emptyset\}$.





\begin{proposition}\label{non-bossy agreeable example}
Any agreeable rule induced by a collection of nested coalitions is non-bossy.
\end{proposition}
\begin{proof}
    Let $\mathcal{C}$ be a collection of nested coalitions and let $\varphi$ be the agreeable rule induced by $\mathcal{S}^\mathcal{C}$. It is clear that $\varphi$ is \emph{peaks-only}. We will prove that the refined agreeable selection   $S^\circ$ associated with $\varphi$ satisfies condition \eqref{non-bossy2} of Proposition \ref{prop non-bossy}. 

Let $R \in  \mathcal{U}^n$, $i \in N$ and $R'_i \in  \mathcal{U}$. 
Assume that $i\notin \mathcal{S}^\circ(R)$ and $\mathcal{S}^\circ(R)\neq \mathcal{S}^\circ(R'_i, R_{-i})$. Suppose by contradiction that $i\notin \mathcal{S}^\circ(R'_i, R_{-i})$. 
Let $\mathcal{S}^\circ(R)=C_t$ and  
$\mathcal{S}^\circ(R'_i, R_{-i})=C_{t'}$. 
As $\mathcal{S}^\circ(R)\neq \mathcal{S}^\circ(R')$, we can assume w.l.o.g that there is $j\in \mathcal{S}^\circ(R)\setminus \mathcal{S}^\circ(R'_i, R_{-i})$. 
Then $j\neq i$ and $j \in C_t$. 
Therefore, $j \notin C_{t'}$ and so $t<t'$, implying $C_{t'}\subsetneq C_t$. However, as $i\notin C_t$ and $C_t\in \mathcal{A}(R)$, we have that
$C_t\in\mathcal{A}(R'_i, R_{-i})$, and we get a contradiction to the fact that  $\mathcal{S}^\circ(R'_i, R_{-i})=C_{t'}$ and $t<t'$.
\end{proof}

\section{Further positive results}\label{further results}

In some situations,  it is more appropriate to assume that instead of a social endowment $\Omega \in \mathbb{R}_+$, there is a profile  $\omega=(\omega_i)_{i \in N} \in \mathbb{R}_+^n$  where, for each $i \in N$, $\omega_i$ denotes agent $i$'s individual endowment of the non-disposable commodity. 
An economy in this context consists of a profile of preferences and a profile of endowments $(R,\omega)$. Let $\mathcal{E}$ denote the domain of all economies with preferences in $\mathcal{U}$.    
In such setting, it is natural to replace the equal division guarantee with the following property.  

\vspace{5 pt}
\noindent
\textbf{Endowments guarantee:} For each  $(R,\omega) \in \mathcal{E}$ and each $i \in N$ such that $t(R_i)=\omega_i$, we have $\varphi_i(R)=\omega_i$.
\vspace{5 pt}

\noindent We can define a \textbf{$\boldsymbol{\omega}$-agreeable selection} as a function $\mathcal{S}^\omega: \mathcal{E} \longrightarrow X$ that, for each $(R,\omega) \in \mathcal{E}$, selects a coalition $\mathcal{S}^\omega(R,\omega)$ in the set
$$\left\{S\subseteq N:\sum_{j\in S}t(R_j)=\sum_{j\in S}\omega_j\right\},$$
and its corresponding \textbf{agreeable reallocation rule} as one that assigns (i) to each agent in $S^\omega(R, \omega)$ their peak amount, and (ii) to each remaining agent their endowment.  Then, for economies with individual endowments, the following variant of Theorem \ref{characterization} is obtained.

\begin{proposition} \label{characterization bis}
 A rule satisfies unanimity, non-obvious manipulability, own-peak-onliness, and the endowments guarantee if and only if it is an agreeable reallocation rule. 
\end{proposition}


Going back to the model with a social endowment, we now present a requirement stating that if the peak of one agent is greater than or equal to the peak of another agent, then the first agent must be assigned an amount greater than or equal to that of the second agent \citep{bochet2024preference}. 


\vspace{5 pt}
\noindent
\textbf{Peak order preservation:} For each  $R \in \mathcal{U}^n$ and each $\{i,j\} \subseteq N$ with $i\neq j$, $t(R_i) \leq t(R_j)$ implies  $\varphi_i(R)\leq \varphi_j(R).$
\vspace{5 pt}






\begin{lemma}
    A rule that satisfies non-obvious manipulability, own-peak-onliness, and peak order preservation meets the equal division guarantee.
\end{lemma}
\begin{proof}
    Let $\varphi$ be a non-obviously manipulable, own-peak-only, and peak responsive rule. Assume, by contradiction, that there are $R \in \mathcal{U}^n$ and  $i\in N$  such that $t(R_i)=\frac{\Omega}{n}$ and $\varphi_i(R)\neq\frac{\Omega}{n}$. Consider the case  $\varphi_i(R)<\frac{\Omega}{n}$ (the other one is symmetric). Let $R_i'\in \mathcal{U}$ be such that $p(R'_i)=\infty$. By \emph{peak order preservation}, for each $\widetilde{R}_{-i} \in \mathcal{U}^{n-1}$ we have that $\varphi_i(R'_i,\widetilde{R}_{-i})\geq \varphi_j(R'_i,\widetilde{R}_{-i})$ for each $j\in N\setminus\{i\}$. Then, by feasibility, $\varphi_i(R'_i,R_{-i})\geq \frac{\Omega}{n}$. By \emph{own-peak-onliness}, we can assume that $R_i$ is such that   $\varphi_i(R'_i,R_{-i}) \ P_i \ \varphi_i(R)$. Then, $R'_i$ is a manipulation of $\varphi$ at $R_i$. Furthermore, if $x\in O^\varphi(R'_i)$, then  $x=\varphi_i(R'_i,R'_{-i})$ for some $R'_{-i} \in \mathcal{U}^{n-1}$. Then, by \emph{peak order preservation} and feasibility, $x\geq \frac{\Omega}{n}$. Therefore, by \emph{own-peak-onliness}, we can assume that $x\mathrel{R_i}\varphi_i(R)$, which implies that  $R'_i$ is an obvious manipulation of $\varphi$ at $R_i$, which is a contradiction.
\end{proof}

In light of the previous lemma, we following result can easily be derived.



\begin{proposition} \label{characterization3}
 A rule 
 satisfies unanimity, non-obvious manipulability,  own-peak-onliness, and peak order preservation if and only if it is a peak order preserving and agreeable rule.\footnote{Notice that the \emph{equal division guarantee} is not invoked in this proposition.} 
\end{proposition}

\begin{remark}  \em In order to construct a peak order preserving and agreeable rule, it is sufficient to consider a peak order preserving agreeable selection, i.e.,  an agreeable selection $\mathcal{S}$ 
such that for each profile $R$ and each $\{i,j\} \subset N$ with $i \neq j$ and $t(R_i) \leq t(R_j)$, the fact that $i\in \mathcal{S}(R)$ implies that $j\in \mathcal{S}(R)$. In particular, the egalitarian status quo rule is a \emph{peak order preserving} agreeable rule.  
\end{remark}

\begin{table}[ht] 
\small
\centering 
\begin{threeparttable}
\begin{tabular}{|l|c|c|c|c||c|c|c|c|}
\hline
& \multicolumn{4}{c||}{Impossibilities} & \multicolumn{4}{c|}{Possibilities}\\
\cline{2-9} 
  & Th. \ref{Impossibility: eff sp and edlb}  &  Fact \ref{fact} & Th. \ref{Impossibility: eff NOM and OPO} & Th. \ref{Impossibility: eff OPO and non-bossy} & Th. \ref{characterization} & Prop. \ref{prop non-bossy} & Prop. \ref{characterization bis} & Prop. \ref{characterization3}
  \\
 \hline \hline
Efficiency & $+$ & $+$ & $+$ & $+$ & & & &\\
\hline 
Unanimity &  &  &  &  & $+$ & $+$ & $+$ & $+$ \\
\hline 
Strategy-proofness & $+$ & &  & & & & &\\
\hline 
Non-obvious manipulability & & $+$& $+$ &  & $+$ & $+$ & $+$ & $+$\\
\hline 
Peaks-onliness &  &  &  &  & & $+$ & &\\
\hline
Own-peak-onliness &  & $+$ & $+$ & $+$ & $+$ & & $+$ &$+$\\
\hline
Equal division lower bound & $+$ &  &  &  & & & &  \\
\hline
Equal division guarantee & & $+$ &  &  & $+$ & $+$ & &  \\
\hline 
Non-bossiness & &  &  & $+$ & & $+$ & & \\
\hline
Non-dictatorship & &  & $+$ &   & & & & \\
\hline

Endowments guarantee & & &  &  & & & $+$ &\\
\hline 
Peak order preservation &  & & & & & & & $+$\\
\hline 

\end{tabular}
\end{threeparttable}
\caption{\emph{(Im)possibilities for non-obviously manipulable rules.}}\label{tabla caracterizaciones}
\end{table}

\section{Final remarks}\label{final remarks}

Our study shows how extending the analysis of \emph{non-obvious manipulability} to a general preference domain clarifies the trade-offs between efficiency, fairness, and behavioral simplicity. Table~\ref{tabla caracterizaciones} summarizes the main impossibility and possibility results, which together outline the conceptual boundaries of \emph{non-obviously manipulable} allotment rules.


Within the domain of single-peaked preferences \cite{arribillaga2025not} characterize an extensive class of rules satisfying non-obvious manipulability, efficiency, own-peak-onliness, and the equal-division guarantee. Such rules are called \textit{simple}. In economies with excess demand, simple rules fully satiate agents whose peak amount is less than or equal to equal division and assign, to each remaining agent, an amount between equal division and their peak. In economies with excess supply, simple rules are defined symmetrically.
If we compare Theorem \ref{characterization} with the characterization in \cite{arribillaga2025not} (Theorem 2), we can see that we have generalized the domain while weakening efficiency to unanimity. A natural question that arises, then, is whether there exists an agreeable rule that is, at the same time, simple and can appear in both characterizations; such rules would be particularly noteworthy. The answer is negative, since we can construct a profile with excess demand in which an agent who requests less than the equal division does not belong to any agreeable coalition. In such a profile, this agent must receive the equal division under any agreeable rule and their peak allocation under any simple rule.


\bibliographystyle{ecta}
\bibliography{biblio}

\begin{thebibliography}{24}
\newcommand{\enquote}[1]{``#1''}
\expandafter\ifx\csname natexlab\endcsname\relax\def\natexlab#1{#1}\fi

\bibitem[\protect\citeauthoryear{Arribillaga and Bonifacio}{Arribillaga and Bonifacio}{2024}]{arribillaga2024obvious}
\textsc{Arribillaga, R.~P. and A.~G. Bonifacio} (2024): \enquote{Obvious manipulations of tops-only voting rules,} \emph{Games and Economic Behavior}, 143, 12--24.

\bibitem[\protect\citeauthoryear{Arribillaga and Bonifacio}{Arribillaga and Bonifacio}{2025{\natexlab{a}}}]{arribillaga2025not}
---\hspace{-.1pt}---\hspace{-.1pt}--- (2025{\natexlab{a}}): \enquote{Not obviously manipulable allotment rules,} \emph{Economic Theory}, 80, 355--380.

\bibitem[\protect\citeauthoryear{Arribillaga and Bonifacio}{Arribillaga and Bonifacio}{2025{\natexlab{b}}}]{arribillaga2025obvious}
---\hspace{-.1pt}---\hspace{-.1pt}--- (2025{\natexlab{b}}): \enquote{Obvious manipulations, consistency, and the uniform rule,} \emph{Economics Letters}, 112344.

\bibitem[\protect\citeauthoryear{Arribillaga and Pepa~Risma}{Arribillaga and Pepa~Risma}{2025{\natexlab{a}}}]{arribillaga2025note}
\textsc{Arribillaga, R.~P. and E.~Pepa~Risma} (2025{\natexlab{a}}): \enquote{A Note on Obvious Manipulations of Quantile Stable Mechanisms: RP Arribillaga, E. Pepa Risma,} \emph{Social Choice and Welfare}, 1--7.

\bibitem[\protect\citeauthoryear{Arribillaga and Pepa~Risma}{Arribillaga and Pepa~Risma}{2025{\natexlab{b}}}]{arribillaga2025obvious2}
---\hspace{-.1pt}---\hspace{-.1pt}--- (2025{\natexlab{b}}): \enquote{Obvious manipulations in matching with and without contracts,} \emph{Games and Economic Behavior}, 151, 70--81.

\bibitem[\protect\citeauthoryear{Aziz and Lam}{Aziz and Lam}{2021}]{aziz2021obvious}
\textsc{Aziz, H. and A.~Lam} (2021): \enquote{Obvious manipulability of voting rules,} in \emph{International Conference on Algorithmic Decision Theory}, Springer, 179--193.

\bibitem[\protect\citeauthoryear{Barber{\`a}, Jackson, and Neme}{Barber{\`a} et~al.}{1997}]{barbera1997strategy}
\textsc{Barber{\`a}, S., M.~O. Jackson, and A.~Neme} (1997): \enquote{Strategy-proof allotment rules,} \emph{Games and Economic Behavior}, 18, 1--21.

\bibitem[\protect\citeauthoryear{Barber{\`a} and Peleg}{Barber{\`a} and Peleg}{1990}]{barbera1990strategy}
\textsc{Barber{\`a}, S. and B.~Peleg} (1990): \enquote{Strategy-proof voting schemes with continuous preferences,} \emph{Social Choice and Welfare}, 7, 31--38.

\bibitem[\protect\citeauthoryear{Bochet, Sakai, and Thomson}{Bochet et~al.}{2024}]{bochet2024preference}
\textsc{Bochet, O., T.~Sakai, and W.~Thomson} (2024): \enquote{Preference manipulations lead to the uniform rule,} \emph{Journal of Economic Theory}, 220, 105879.

\bibitem[\protect\citeauthoryear{Chatterji, Mass{\'o}, and Serizawa}{Chatterji et~al.}{2025}]{chatterji2025strategy}
\textsc{Chatterji, S., J.~Mass{\'o}, and S.~Serizawa} (2025): \enquote{On strategy-proofness and the salience of single-peakedness in a private goods allotment problem,} \emph{Games and Economic Behavior}, 150, 48--70.

\bibitem[\protect\citeauthoryear{Ching}{Ching}{1994}]{ching1994alternative}
\textsc{Ching, S.} (1994): \enquote{An alternative characterization of the uniform rule,} \emph{Social Choice and Welfare}, 11, 131--136.

\bibitem[\protect\citeauthoryear{Gong, Dietzenbacher, and Peters}{Gong et~al.}{2024}]{gong2024mechanisms}
\textsc{Gong, D., B.~Dietzenbacher, and H.~Peters} (2024): \enquote{Mechanisms and axiomatics for division problems with single-dipped preferences,} \emph{Economic Theory}, 78, 789--813.

\bibitem[\protect\citeauthoryear{Hurwicz}{Hurwicz}{1972}]{hurwicz1972informationally}
\textsc{Hurwicz, L.} (1972): \enquote{On informationally decentralized systems,} \emph{Decision and organization: A volume in Honor of J. Marschak}.

\bibitem[\protect\citeauthoryear{Klaus, Peters, and Storcken}{Klaus et~al.}{1997}]{klaus1997strategy}
\textsc{Klaus, B., H.~Peters, and T.~Storcken} (1997): \enquote{Strategy-proof division of a private good when preferences are single-dipped,} \emph{Economics Letters}, 55, 339--346.

\bibitem[\protect\citeauthoryear{Momi}{Momi}{2013}]{momi2013note}
\textsc{Momi, T.} (2013): \enquote{Note on social choice allocation in exchange economies with many agents,} \emph{Journal of Economic Theory}, 148, 1237--1254.

\bibitem[\protect\citeauthoryear{Ortega and Segal-Halevi}{Ortega and Segal-Halevi}{2022}]{ortega2022obvious}
\textsc{Ortega, J. and E.~Segal-Halevi} (2022): \enquote{Obvious manipulations in cake-cutting,} \emph{Social Choice and Welfare}, 1--20.

\bibitem[\protect\citeauthoryear{Psomas and Verma}{Psomas and Verma}{2022}]{psomas2022fair}
\textsc{Psomas, A. and P.~Verma} (2022): \enquote{Fair and efficient allocations without obvious manipulations,} \emph{arXiv preprint arXiv:2206.11143}.

\bibitem[\protect\citeauthoryear{Serizawa}{Serizawa}{2002}]{serizawa2002inefficiency}
\textsc{Serizawa, S.} (2002): \enquote{Inefficiency of strategy-proof rules for pure exchange economies,} \emph{Journal of Economic Theory}, 106, 219--241.

\bibitem[\protect\citeauthoryear{Sprumont}{Sprumont}{1991}]{sprumont1991division}
\textsc{Sprumont, Y.} (1991): \enquote{The division problem with single-peaked preferences: a characterization of the uniform allocation rule,} \emph{Econometrica}, 509--519.

\bibitem[\protect\citeauthoryear{Thomson}{Thomson}{1994{\natexlab{a}}}]{thomson1994consistent}
\textsc{Thomson, W.} (1994{\natexlab{a}}): \enquote{Consistent solutions to the problem of fair division when preferences are single-peaked,} \emph{Journal of Economic Theory}, 63, 219--245.

\bibitem[\protect\citeauthoryear{Thomson}{Thomson}{1994{\natexlab{b}}}]{thomson1994resource}
---\hspace{-.1pt}---\hspace{-.1pt}--- (1994{\natexlab{b}}): \enquote{Resource-monotonic solutions to the problem of fair division when preferences are single-peaked,} \emph{Social Choice and Welfare}, 11, 205--223.

\bibitem[\protect\citeauthoryear{Thomson}{Thomson}{1995}]{thomson1995population}
---\hspace{-.1pt}---\hspace{-.1pt}--- (1995): \enquote{Population-monotonic solutions to the problem of fair division when preferences are single-peaked,} \emph{Economic Theory}, 5, 229--246.

\bibitem[\protect\citeauthoryear{Troyan and Morrill}{Troyan and Morrill}{2020}]{troyan2020obvious}
\textsc{Troyan, P. and T.~Morrill} (2020): \enquote{Obvious manipulations,} \emph{Journal of Economic Theory}, 185, 104970.

\bibitem[\protect\citeauthoryear{Zhou}{Zhou}{1991}]{zhou1991inefficiency}
\textsc{Zhou, L.} (1991): \enquote{Inefficiency of strategy-proof allocation mechanisms in pure exchange economies,} \emph{Social Choice and Welfare}, 8, 247--254.

\end{thebibliography}

\end{document}